\documentclass[12pt]{amsart}
\pdfoutput=1

\usepackage[paperheight=11in, paperwidth=8.5in, left=1in, top=1.25in, right=1in, bottom=1.25in]{geometry}
\usepackage[linktocpage=true, linktoc=all, colorlinks=true, linkcolor=black, citecolor=black, filecolor=black, urlcolor=black, pagebackref=false, pdfstartpage={1}, pdfstartview={FitH}, pdftitle={}, pdfauthor={}, pdfsubject={}, pdfcreator={}, pdfproducer={}, pdfkeywords={}]{hyperref}
\usepackage{afterpage}
\usepackage{amsfonts}
\usepackage{amsmath}
\allowdisplaybreaks[4]
\usepackage{amssymb}
\usepackage{amsthm}
\usepackage{array}
\usepackage{bbm}
\usepackage[justification=centering]{caption}
\usepackage[capitalize,nameinlink,noabbrev,compress]{cleveref}

\usepackage{enumerate}
\usepackage{enumitem}
\usepackage{float}
\restylefloat{figure}
\usepackage{graphicx}
\usepackage{mathrsfs}
\usepackage[framemethod=tikz,ntheorem,xcolor]{mdframed}
\usepackage{parcolumns}
\usepackage{scalerel}
\usepackage{tabularx}
\usepackage{tikz}\pdfpageattr{/Group <</S /Transparency /I true /CS /DeviceRGB>>}
\usetikzlibrary{arrows,calc,cd,intersections,matrix,positioning,through}
\tikzset{commutative diagrams/.cd,every label/.append style = {font = \normalsize}}
\usepackage[all]{xy}
\usepackage[aligntableaux=center]{ytableau}

\numberwithin{equation}{section}
\newtheorem{thm}{Theorem}
\AfterEndEnvironment{thm}{\noindent\ignorespaces}
\numberwithin{thm}{section}

\AfterEndEnvironment{conj}{\noindent\ignorespaces}
\newtheorem{cor}[thm]{Corollary}
\AfterEndEnvironment{cor}{\noindent\ignorespaces}
\newtheorem{lem}[thm]{Lemma}
\AfterEndEnvironment{lem}{\noindent\ignorespaces}
\newtheorem{prob}[thm]{Problem}
\AfterEndEnvironment{prob}{\noindent\ignorespaces}
\newtheorem{prop}[thm]{Proposition}
\AfterEndEnvironment{prop}{\noindent\ignorespaces}

\theoremstyle{definition}
\newtheorem{defn}[thm]{Definition}
\AfterEndEnvironment{defn}{\noindent\ignorespaces}
\newtheorem{eg_no_qed}[thm]{Example}
\newenvironment{eg}[1][]{\begin{eg_no_qed}[#1]\pushQED{\qed}}{\popQED\end{eg_no_qed}}
\AfterEndEnvironment{eg}{\noindent\ignorespaces}
\newtheorem{rmk}[thm]{Remark}
\AfterEndEnvironment{rmk}{\noindent\ignorespaces}

\theoremstyle{remark}

\AfterEndEnvironment{claim}{\noindent\ignorespaces}
\newtheorem*{claimpf_no_qed}{Proof of Claim}

\AfterEndEnvironment{claimpf}{\noindent\ignorespaces}

\DeclareMathOperator{\ad}{ad}
\newcommand{\adjoint}[1]{#1 ^*}

\DeclareMathOperator{\diag}{diag}

\renewcommand{\eqref}[1]{\hyperref[#1]{\textup{(\ref*{#1})}}}

\DeclareMathOperator{\Fl}{Fl}

\DeclareMathOperator{\GL}{GL}
\newcommand{\gl}{\mathfrak{gl}}
\DeclareMathOperator{\Gr}{Gr}
\DeclareMathOperator{\grad}{grad}

\DeclareMathOperator{\im}{im}

\DeclareMathOperator{\spn}{span}

\DeclareMathOperator{\tr}{tr}

\newcommand{\B}{\operatorname{B}}

\newcommand{\bfnu}{\nu}
\newcommand{\bflambda}{\lambda}

\newcommand{\Bminus}{\operatorname{B}^{-}}

\newcommand{\cyc}[2]{\mathsf{cyc}(#1;#2)}
\renewcommand{\dashrightarrow}{\mathrel{\ThisStyle{\ooalign{$\SavedStyle\rightarrow$\cr\hfil\textcolor{white}{\rule{2\LMpt}{1\LMex}}\kern2\LMpt\hfil}}}}
\newcommand{\Diag}[1]{\hspace*{1pt}\mathsf{Diag}(#1)\hspace*{1pt}}

\newcommand{\eval}[1]{\hspace*{-1pt}\big\rvert_{#1}\hspace*{2pt}}

\renewcommand{\H}{\operatorname{T}}

\newcommand{\hypersimplex}{\mathsf{HS}}
\newcommand{\I}{I}
\newcommand{\ii}{\mathrm{i}\hspace*{0.5pt}}

\newcommand{\Jac}{\mathcal{J}}

\DeclareMathOperator{\Kterm}{\pi_{U}}
\DeclareMathOperator{\kterm}{\pi_{\mathfrak{u}}}
\newcommand{\Kry}[2]{\mathsf{Vand}(#1,#2)}

\newcommand{\Moser}{\mathsf{Mos}}
\newcommand{\Moserextend}{\widetilde{\mathsf{Mos}}}
\newcommand{\N}{\operatorname{N}}

\newcommand{\Orbit}{\mathcal{O}}

\newcommand{\Perm}[1]{\mathsf{Perm}(#1)}

\newcommand{\twist}{\vartheta}
\newcommand{\twistorbit}{\vartheta_{\bflambda}}
\newcommand{\U}{\operatorname{U}}
\newcommand{\uu}{\mathfrak{u}}

\title{Symmetric Toda, gradient flows, and tridiagonalization}

\author{Anthony M. Bloch}
\address{Department of Mathematics, University of Michigan}
\email{\href{mailto:abloch@umich.edu}{abloch@umich.edu}}
\author{Steven N. Karp}
\address{Department of Mathematics, University of Notre Dame}
\email{\href{mailto:skarp2@nd.edu}{skarp2@nd.edu}}

\thanks{A.M.B.\ was partially supported by NSF grants DMS-1613819 and DMS-2103026, and AFOSR grant FA 9550-22-1-0215.}

\begin{document}

\begin{abstract}
The Toda lattice (1967) is a Hamiltonian system given by $n$ points on a line governed by an exponential potential. Flaschka (1974) showed that the Toda lattice is integrable by interpreting it as a flow on the space of symmetric tridiagonal $n\times n$ matrices, while Moser (1975) showed that it is a gradient flow on a projective space. The symmetric Toda flow of Deift, Li, Nanda, and Tomei (1986) generalizes the Toda lattice flow from tridiagonal to all symmetric matrices. They showed the flow is integrable, in the classical sense of having $d$ integrals in involution on its $2d$-dimensional phase space. The system may be viewed as integrable in other ways as well. Firstly, Symes (1980, 1982) solved it explicitly via $QR$-factorization and conjugation. Secondly, Deift, Li, Nanda, and Tomei (1986) `tridiagonalized' the system into a family of tridiagonal Toda lattices which are solvable and integrable. In this paper we derive their tridiagonalization procedure in a natural way using the fact that the symmetric Toda flow is diffeomorphic to a twisted gradient flow on a flag variety, which may then be decomposed into flows on a product of Grassmannians. These flows may in turn be embedded into projective spaces via Pl\"{u}cker embeddings, and mapped back to tridiagonal Toda lattice flows using Moser's construction. In addition, we study the tridiagonalized flows projected onto a product of permutohedra, using the twisted moment map of Bloch, Flaschka, and Ratiu (1990). These ideas are facilitated in a natural way by the theory of total positivity, building on our previous work (2023).
\end{abstract}

\dedicatory{In honor of Hermann Flaschka}
\maketitle

\section{Introduction}\label{sec_introduction}

\noindent This paper concerns the symmetric Toda flows and their connections with the classical Toda lattice, gradient flows on adjoint orbits, and flows on moment polytopes. The (finite nonperiodic) {\itshape Toda lattice} \cite{toda67b} (cf.\ \cite{kodama_shipman18}) is a Hamiltonian system of $n$ points on a line of unit mass governed by an exponential potential, with Hamiltonian
$$
\frac{1}{2}\sum_{i=1}^np_i^2 + \sum_{i=1}^{n-1}e^{q_i - q_{i+1}}.
$$
Following Flaschka \cite{flaschka74}, we make the change of variables
$$
a_i := \textstyle\frac{1}{2}e^{\frac{q_{i}-q_{i+1}}{2}} \; \text{ for } 1 \le i \le n-1 \quad \text{ and } \quad b_i :=  -\frac{1}{2}p_i \; \text{ for } 1 \le i \le n,
$$
which we arrange into the symmetric tridiagonal matrix
$$
M := \begin{bmatrix}
b_1 & a_1 & 0 & \cdots & 0 \\
a_1 & b_2 & a_2 & \cdots & 0 \\
0 & a_2 & b_3 & \cdots & 0 \\
\vdots & \vdots & \vdots & \ddots & \vdots \\
0 & 0 & 0 & \cdots & b_n
\end{bmatrix}; \quad \text{ we also let } \kterm(M) := \begin{bmatrix}
0 & -a_1 & 0 & \cdots & 0 \\
a_1 & 0 & -a_2 & \cdots & 0 \\
0 & a_2 & 0 & \cdots & 0 \\
\vdots & \vdots & \vdots & \ddots & \vdots \\
0 & 0 & 0 & \cdots & 0
\end{bmatrix}
$$
be the skew-symmetric part of $M$. Then the Hamiltonian equations can be written in Lax form as
\begin{align}\label{toda_flaschka_symmetric}
\dot{M} = [M, \kterm(M)],
\end{align}
where $\dot{M}$ denotes the derivative of $M$ with respect to time $t$. Then the eigenvalues of $M$ are preserved along the flow, and allow one to define $n$ integrals in involution \cite{moser75,symes80,symes82}, showing that the Toda lattice is integrable.

Moreover, Moser \cite{moser75} expressed the Toda lattice flow as a gradient flow on a projective space. Namely, let $\lambda_1 > \cdots > \lambda_n$ denote the eigenvalues of $M$ (they are necessarily distinct), and for $1 \le i \le n$, let $v^{(i)}\in\mathbb{R}^n$ be the eigenvector of $M$ with eigenvalue $\lambda_i$ such that $\|v^{(i)}\|=1$ and its first entry $v^{(i)}_1$ is positive ($v^{(i)}_1$ is necessarily nonzero). Setting $x := (v^{(1)}_1 : \cdots : v^{(n)}_1)\in\mathbb{P}^{n-1}(\mathbb{R})$, we can write the Hamiltonian equations \eqref{toda_flaschka_symmetric} as
\begin{align}\label{moser_gradient_intro}
\dot{x}_i = \lambda_ix_i \quad \text{ for all } 1 \le i \le n.
\end{align}
Conversely, given $x\in\mathbb{P}^{n-1}(\mathbb{R})$ with positive entries, Moser constructed a unique corresponding symmetric tridiagonal matrix $M$ with eigenvalues $\lambda_1 > \cdots > \lambda_n$ and positive entries immediately above and below the diagonal, giving a diffeomorphism
\begin{align}\label{moser_map_intro}
x \mapsto M.
\end{align} 

More generally, Deift, Li, Nanda, and Tomei \cite{deift_li_nanda_tomei86} considered the flow \eqref{toda_flaschka_symmetric} for arbitrary (not necessarily tridiagonal) $n\times n$ symmetric matrices $M$, called the {\itshape (full) symmetric Toda flow}. They constructed $\lfloor\frac{n^2}{4}\rfloor$ integrals in involution, showing that the flow is integrable.

This system may be seen to be integrable in other ways as well. Firstly, Symes \cite{symes80,symes82} found an explicit solution using factorization and conjugation. Namely, given an $n\times n$ matrix $g$, let $\Kterm(g)$ denote the unitary matrix obtained by applying Gram--Schmidt orthonormalization to the columns of $g$; equivalently, $\Kterm(g)$ is the $Q$-term in the $QR$-factorization of $g$. Then if $M(t)$ denotes the solution to \eqref{toda_flaschka_symmetric} beginning at the symmetric matrix $M_0$, we have
$$
M(t) = \Kterm(\exp(tM_0))^{-1}M_0\Kterm(\exp(t M_0)) \quad \text{ for all } t\in\mathbb{R}.
$$

Secondly, Deift, Li, Nanda, and Tomei \cite[Section 7]{deift_li_nanda_tomei86} `tridiagonalized' the symmetric Toda flow into a family of $n-1$ Toda lattices which are solvable and integrable. The main goal of this paper is to derive this tridiagonalization procedure in a natural and geometric way.

In order to state our results, we introduce some notation. Let $\U_n$ denote the group of $n\times n$ unitary matrices, and let $\uu_n$ denote its Lie algebra, consisting of all $n\times n$ skew-Hermitian matrices. Given a symmetric (or more generally, Hermitian) $n\times n$ matrix $M$, we associate to it the skew-Hermitian matrix $L := \ii M$ (where $\ii = \sqrt{-1}$). If $M$ has eigenvalues $\bflambda$, then $L$ lies in the {\itshape adjoint orbit} $\Orbit_{\bflambda}$ of $\uu_n$ consisting of all skew-Hermitian matrices with eigenvalues $\ii\lambda_1, \dots, \ii\lambda_n$. Then we can write the symmetric Toda flow as the flow
\begin{align}\label{toda_flaschka_skew}
\dot{L} = [L, \kterm(-\ii L)] \quad \text{ on } \Orbit_{\bflambda}.
\end{align}
Above, the projection $\kterm(\cdot)$ onto $\uu_n$ is defined such that $N - \kterm(N)$ is upper-triangular with real diagonal entries.

In the case that the eigenvalues $\lambda_1, \dots, \lambda_n$ are distinct, we construct a piecewise-smooth involution $\twistorbit$ on $\Orbit_{\bflambda}$, called the {\itshape twist map}. It sends $L = g(\ii\Diag{\bflambda})g^{-1}$ to $g^{-1}(\ii\Diag{\bflambda})g$, where for a given $L$ the unitary matrix $g\in\U_n$ of eigenvectors is chosen according to a certain normalization condition \eqref{canonical_schubert_representative}, coming from the Bruhat decomposition. We then use $\twistorbit$ to show that the symmetric Toda flow \eqref{toda_flaschka_skew} is a gradient flow (see \cref{full_symmetric_toda_gradient}):
\begin{thm}\label{full_symmetric_toda_gradient_intro}
Let $\bflambda = (\lambda_1 > \cdots > \lambda_n)$. Then the symmetric Toda flow \eqref{toda_flaschka_skew} on $\Orbit_{\bflambda}$ is, upon applying the twist map $\twistorbit$, the gradient flow in the K\"{a}hler metric with respect to $\Diag{-\ii\lambda_1, \dots, -\ii\lambda_n}\in\uu_n$.
\end{thm}

Now let $\Fl_n(\mathbb{C})$ denote the {\itshape complete flag variety}, consisting of all chains of subspaces $(V_1, \dots, V_{n-1})$ of $\mathbb{C}^n$ such that
$$
V_1 \subset \cdots \subset V_{n-1} \quad \text{ and } \quad \dim(V_k) = k \text{ for all } 1 \le k \le n-1.
$$
Given $L\in\Orbit_{\bflambda}$, let $V_k$ (for $1 \le k \le n-1$) denote the subspace of $\mathbb{C}^n$ spanned by the eigenvectors corresponding to the eigenvalues $\ii\lambda_1, \dots, \ii\lambda_k$. It is well-known (and one can verify) that the map
\begin{align}\label{orbit_complete_isomorphism_intro}
\Orbit_{\bflambda}\xrightarrow{\cong}\Fl_n(\mathbb{C}), \quad L \mapsto (V_1, \dots, V_{n-1})
\end{align}
is an isomorphism. Also, $\Fl_n(\mathbb{C})$ has the structure of a projective variety, given by the {\itshape Pl\"{u}cker embedding}
\begin{align}\label{plucker_embedding_intro}
\Fl_n(\mathbb{C})\hookrightarrow\prod_{k=1}^{n-1}\mathbb{P}^{\left(\hspace*{-1pt}\binom{n}{k}-1\right)}(\mathbb{C})
\end{align}
(see \eqref{plucker_embedding_flag}). Then by \cref{full_symmetric_toda_gradient_intro}, we can embed the symmetric Toda flow \eqref{toda_flaschka_skew} on $\Orbit_{\bflambda}$ as a gradient flow on a product of projective spaces. The gradient flow on each projective space can be written in the form \eqref{moser_gradient_intro}, which by Moser's map \eqref{moser_map_intro} is equivalent to a Toda lattice flow; in general, for the $k$th projective space, the associated tridiagonal matrix has size $\binom{n}{k}\times\binom{n}{k}$. We may summarize this as follows (see \cref{tridiagonalization_via_plucker} for a precise statement):
\begin{thm}\label{tridiagonalization_via_plucker_intro}
The tridiagonalization procedure of Deift, Li, Nanda, and Tomei \cite[Section 7]{deift_li_nanda_tomei86} of the symmetric Toda flow \eqref{toda_flaschka_symmetric} on $n\times n$ symmetric matrices $M$ is given as follows. By multiplying by $\ii$ and applying the maps $\twistorbit$, \eqref{orbit_complete_isomorphism_intro}, and \eqref{plucker_embedding_intro}, we can embed the symmetric Toda flow as a gradient flow on the product of projective spaces
$$
\prod_{k=1}^{n-1}\mathbb{P}^{\left(\hspace*{-1pt}\binom{n}{k}-1\right)}(\mathbb{R}).
$$
For $(p_1, \dots, p_{n-1}) \in \prod_{k=1}^{n-1}\mathbb{P}^{\left(\hspace*{-1pt}\binom{n}{k}-1\right)}(\mathbb{R})$, we disregard the zero coordinates of every $p_i$ (for $1 \le i \le n-1$) and replace each nonzero coordinate of $p_i$ with its absolute value. Then applying Moser's map \eqref{moser_map_intro} embeds the symmetric Toda flow into a family of $n-1$ tridiagonal Toda lattice flows.
\end{thm}

As an application, we apply a construction of Bloch, Flaschka, and Ratiu \cite{bloch_flaschka_ratiu90} which maps the tridiagonal Toda lattice flows homeomorphically onto the associated moment polytope (a {\itshape permutohedron}). \cref{tridiagonalization_via_plucker_intro} allows us to embed the symmetric Toda flow as a flow on a product of $n-1$ permutohedra (see \cref{moment_embedding_permutohedron_product}). We also consider a closely related construction, which maps the symmetric Toda flow to a flow on a product of $n-1$ {\itshape hypersimplices} (moment polytopes for Grassmannians). We pose the problem of whether this map is an embedding (see \cref{interpolated_moment_map_problem}).

We mention that one of our motivations for studying the symmetric Toda flow is its relationship with the theory of {\itshape total positivity} for flag varieties, introduced by Lusztig \cite{lusztig94}. While total positivity is not part of the statements of our main results, it was key to our preceding work \cite{bloch_karp23b} and provided the impetus for this work. For example, \cref{full_symmetric_toda_gradient_intro} above is inspired by \cite[Theorem 8.6]{bloch_karp23b}. We refer to \cite{bloch_karp23b} for further details on total positivity, as well as for references to related works in the literature.

\subsection*{Outline}
In \cref{sec_background} we recall some background. In \cref{sec_general_twist} we introduce the twist map $\twistorbit$. In \cref{sec_toda_gradient} we show that the symmetric Toda flow is a twisted gradient flow. In \cref{sec_moser} we discuss the tridiagonal Toda lattice and Moser's map \eqref{moser_map_intro}. In \cref{sec_tridiagonalization} we recall the tridiagonalization construction of \cite{deift_li_nanda_tomei86}, and give our new interpretation of it. In \cref{moment_sec} we apply this construction to study Toda flows on moment polytopes.

We remark that the Toda flows are traditionally studied over the real numbers, and our statement of \cref{tridiagonalization_via_plucker_intro} follows this tradition. However, all of our arguments and results hold over the complex numbers, so we work over $\mathbb{C}$ (for example, \cref{tridiagonalization_via_plucker_intro} follows by specializing \cref{tridiagonalization_via_plucker} over $\mathbb{R}$). This is for the sake both of generality, and for consistency with our Lie-algebraic setup. On the other hand, none of our arguments require working over $\mathbb{C}$ (rather than $\mathbb{R}$).

\subsection*{Acknowledgments}
We thank the anonymous reviewers for helpful comments.

\section{Background}\label{sec_background}

\noindent We recall some important background on flag varieties, adjoint orbits, and total positivity. Our notation throughout this paper is consistent with our previous work \cite{bloch_karp23b}, to which we refer for further details and examples.

Throughout the paper, we fix a strictly decreasing vector $\bflambda = (\lambda_1, \dots, \lambda_n)\in\mathbb{R}^n$. We point out that in \cite{bloch_karp23b}, we more generally consider $\bflambda$ which are weakly decreasing; we will not need to do so here.

Let $\mathbb{N} := \{0, 1, 2, \dots, \}$, and for $n\in\mathbb{N}$, define $[n] := \{1, 2, \dots, n\}$. For $k\in\mathbb{N}$, we let $\binom{[n]}{k}$ denote the set of $k$-element subsets of $[n]$. We let $e_1, \dots, e_n$ denote the unit vectors of $\mathbb{C}^n$. Given an $m\times n$ matrix $A$ and subsets $I\subseteq [m]$ and $J\subseteq [n]$, we let $A_{I,J}$ denote the submatrix of $A$ using rows $I$ and columns $J$. If $|I| = |J|$, we let $\Delta_{I,J}(A)$ denote $\det(A_{I,J})$, called a {\itshape minor} of $A$. If $J = [k]$, where $k = |I|$, we call $\Delta_{I,[k]}(A)$ a {\itshape left-justified minor} of $A$, which we denote by $\Delta_I(A)$.

We let $\mathbb{P}^{n-1}(\mathbb{C})$ denote the {\itshape projective space} of all nonzero vectors $(x_1 : \cdots : x_n)$ modulo rescaling. We let $\gl_n(\mathbb{C})$ denote the Lie algebra of $n\times n$ matrices over $\mathbb{C}$, with Lie bracket
$$
[L,M] := LM - ML \quad \text{ for all } L,M\in\gl_n(\mathbb{C}).
$$
We let $\ad_L := [L,\cdot]$ denote the {\itshape adjoint operator} of $L\in\gl_n(\mathbb{C})$. We let $\diag(L)\in\mathbb{C}^n$ denote the vector of diagonal entries of $L\in\gl_n(\mathbb{C})$. Finally, we let $\Diag{c_1, \dots, c_n}\in\gl_n(\mathbb{C})$ denote the $n\times n$ diagonal matrix with diagonal entries $c_1, \dots, c_n$.

We recall the {\itshape Cauchy--Binet identity} (see e.g.\ \cite[I.(14)]{gantmacher59}): if $A$ is an $m\times n$ matrix, $B$ is an $n\times p$ matrix, and $1 \le k \le m,p$, then
\begin{align}\label{cauchy-binet}
\Delta_{I,J}(AB) = \sum_{K\in\binom{[n]}{k}}\Delta_{I,K}(A)\Delta_{K,J}(B) \quad \text{ for all } I\in\textstyle\binom{[m]}{k} \text{ and } J\in\binom{[p]}{k}.
\end{align}

\subsection{Flag varieties and adjoint orbits}\label{sec_background_flag}
We introduce Grassmannians and complete flag varieties, which will play an important role in the paper.
\begin{defn}\label{defn_grassmannian}
Let $0 \le k \le n$. We define the {\itshape Grassmannian} $\Gr_{k,n}(\mathbb{C})$ as the set of all $k$-dimensional linear subspaces of $\mathbb{C}^n$. Given $V\in\Gr_{k,n}(\mathbb{C})$, we say that an $n\times k$ matrix $A$ {\itshape represents} $V$ if its columns form a basis of $V$. We have the {\itshape Pl\"{u}cker embedding}
\begin{align}\label{plucker_embedding_grassmannian}
\Gr_{k,n}(\mathbb{C})\hookrightarrow\mathbb{P}^{\left(\hspace*{-1pt}\binom{n}{k}-1\right)}(\mathbb{C}), \quad V \mapsto \big(\Delta_I(A) : I\in\textstyle\binom{[n]}{k}\big),
\end{align}
which does not depend on the choice of $A$. We call the projective coordinates $\Delta_I(\cdot)$ on $\Gr_{k,n}(\mathbb{C})$ {\itshape Pl\"{u}cker coordinates}.
\end{defn}

\begin{defn}\label{defn_Fl}
Let $\GL_n(\mathbb{C})$ denote the general linear group of all $n\times n$ invertible matrices over $\mathbb{C}$, and let $\B_n(\mathbb{C})$ denote the Borel subgroup of $\GL_n(\mathbb{C})$ of all upper-triangular matrices. We define the {\itshape complete flag variety} as the quotient
$$
\Fl_n(\mathbb{C}) := \GL_n(\mathbb{C})/\B_n(\mathbb{C}),
$$
which we may identify with the variety of complete flags of linear subspaces of $\mathbb{C}^n$
$$
\{V = (V_1, \dots, V_{n-1}) : 0 \subset V_1 \subset \cdots \subset V_{n-1} \subset \mathbb{C}^n \text{ and } \dim(V_k) = k \text{ for } 1 \le k \le n-1\}.
$$
This identification sends $g\in\GL_n(\mathbb{C})/\B_n(\mathbb{C})$ to the tuple $(V_1, \dots, V_{n-1})$, where each $V_k$ is the span of the first $k$ columns of $g$. We will freely alternate between regarding complete flags as elements $g\in\GL_n(\mathbb{C})/\B_n(\mathbb{C})$ or as tuples $(V_1, \dots, V_{n-1})$.

We have the {\itshape Pl\"{u}cker embedding}
\begin{align}\label{plucker_embedding_flag}
\Fl_n(\mathbb{C})\hookrightarrow\prod_{k=1}^{n-1}\mathbb{P}^{\left(\hspace*{-1pt}\binom{n}{k}-1\right)}(\mathbb{C}), \quad g \mapsto \big(\Delta_I(g) : I\in\textstyle\binom{[n]}{k}\big)_{1 \le k \le n-1},
\end{align}
which is given by the embedding
\begin{align}\label{embedding_Fl_Gr}
\Fl_n(\mathbb{C}) \hookrightarrow \prod_{k=1}^{n-1}\Gr_{k,n}(\mathbb{C}), \quad (V_1, \dots, V_{n-1}) \mapsto (V_1, \dots, V_{n-1}),
\end{align}
and then applying the Pl\"{u}cker embedding \eqref{plucker_embedding_grassmannian} to each term $\Gr_{k,n}(\mathbb{C})$. We call the left-justified minors $\Delta_I(g)$ appearing above the {\itshape Pl\"{u}cker coordinates} of $g\in\Fl_n(\mathbb{C})$ (also known as {\itshape flag minors}).
\end{defn}

\begin{eg}\label{eg_Fl}
We can write a generic element of the complete flag variety $\Fl_3(\mathbb{C})$ as
$$
g := \begin{bmatrix}
1 & 0 & 0 \\
a & 1 & 0 \\
b & c & 1
\end{bmatrix}\in\GL_3(\mathbb{C})/\B_3(\mathbb{C}), \quad \text{ where } a,b,c\in\mathbb{C}.
$$
The Pl\"{u}cker embedding \eqref{plucker_embedding_flag} takes $g$ to
\begin{multline*}
\big((\Delta_1(g) : \Delta_2(g) : \Delta_3(g)), (\Delta_{12}(g) : \Delta_{13}(g) : \Delta_{23}(g)\big) \\
= \big((1 : a : b), (1 : c : ac-b)\big) \in \mathbb{P}^3(\mathbb{C})\times\mathbb{P}^3(\mathbb{C}).\qedhere
\end{multline*}
\end{eg}

We recall that $\Orbit_{\bflambda}$ denotes the adjoint orbit of $\uu_n$ consisting of all skew-Hermitian matrices with eigenvalues $\ii\lambda_1, \dots, \ii\lambda_n$:
$$
\Orbit_{\bflambda} = \{g(\ii\Diag{\bflambda})g^{-1} : g\in\U_n\}.
$$
We observe that we can write the isomorphism $\Orbit_{\bflambda}\xrightarrow{\cong}\Fl_n(\mathbb{C})$ from \eqref{orbit_complete_isomorphism_intro} equivalently as
$$
\Orbit_{\bflambda}\xrightarrow{\cong}\Fl_n(\mathbb{C}), \quad g(\ii\Diag{\bflambda})g^{-1} \mapsto g.
$$

\begin{rmk}\label{projection_sum}
The embedding \eqref{embedding_Fl_Gr} has a natural interpretation in $\Orbit_{\bflambda}$. Namely, given $(V_1, \dots, V_{n-1})\in\Fl_n(\mathbb{C})$, let $\ii M\in\Orbit_{\bflambda}$ be the corresponding element under the isomorphism \eqref{orbit_complete_isomorphism_intro}. Then
\begin{align}\label{projection_sum_formula}
M = \bigg(\sum_{k=1}^{n-1}(\lambda_k - \lambda_{k+1})P_k\bigg) + \lambda_n\I_n,
\end{align}
where $P_k\in\gl_n(\mathbb{C})$ is orthogonal projection onto $V_k$ \cite[Lemma 4.16]{bloch_karp23b}.
\end{rmk}

\subsection{Total positivity}\label{sec_background_positivity}
We recall the notion of total positivity for $\Fl_n(\mathbb{C})$.
\begin{defn}[\cite{lusztig94,lusztig98}]\label{defn_tnn_Fl}
Let $0 \le k \le n$. We say that an element of $\Fl_n(\mathbb{C})$ is {\itshape totally positive} (respectively, {\itshape totally nonnegative}) if all its Pl\"{u}cker coordinates are real and positive (respectively, nonnegative), up to rescaling. This defines the totally positive part $\Fl_n^{>0}$ and the totally nonnegative part $\Fl_n^{\ge 0}$. (This definition is different from, but equivalent to, the original definition of Lusztig \cite{lusztig94,lusztig98}; see \cite[Section 1.4]{bloch_karp23} for references and a history of this equivalence.)

We define the {\itshape totally positive} part $\Orbit_{\bflambda}^{>0}$ to be the inverse image of $\Fl_n^{>0}$ under the isomorphism \eqref{orbit_complete_isomorphism_intro}. We similarly define the {\itshape totally nonnegative} part $\Orbit_{\bflambda}^{\ge 0}$.
\end{defn}

\begin{eg}\label{eg_tnn_Fl}
The element $g\in\Fl_3(\mathbb{C})$ from \cref{eg_Fl} is totally positive if and only if $a, b, c, ac-b>0$.
\end{eg}

\section{The general twist map}\label{sec_general_twist}

\noindent In this section we construct an involution $\twist$ on $\Fl_n(\mathbb{C})$, which we call the {\itshape twist map}. This generalizes the totally nonnegative twist map we defined on $\Fl_n^{\ge 0}$ in \cite{bloch_karp23b}; see \cref{general_vs_positive_twist}. We refer to \cite[Section 3.3]{bloch_karp23b} for further motivation and a discussion of related work.

We begin by recalling the Bruhat decomposition of $\Fl_n(\mathbb{C})$; for further details, see, e.g., \cite[Section 1.2]{brion05}.
\begin{defn}\label{defn_schubert_cell}
Given $n\in\mathbb{N}$, let $\mathfrak{S}_n$ denote the symmetric group of all permutations of $[n]$. For $w\in\mathfrak{S}_n$, we define the {\itshape (signed) permutation matrix} $\mathring{w}\in\GL_n(\mathbb{C})$ by $$
\mathring{w}_{i,j} := \begin{cases}
\pm 1, & \text{ if $i = w(j)$}; \\
0, & \text{ otherwise},
\end{cases}\quad \text{ for } 1 \le i,j \le n,
$$
where the signs are chosen so that all left-justified minors of $\mathring{w}$ are nonnegative. Note that
\begin{align}\label{permutation_matrix_inverse}
\mathring{(w^{-1})} = \delta_n(\mathring{w})^{-1}\delta_n, \quad \text{ where } \delta_n := \Diag{1, -1, 1, \dots, (-1)^{n-1}}.
\end{align}
We also regard $\mathring{w}$ as an element of $\Fl_n(\mathbb{C})$, and define the {\itshape Schubert cell}
$$
\mathring{X}^w := \B_n(\mathbb{C})\cdot\mathring{w}\subseteq\Fl_n(\mathbb{C}),
$$
which consists of all $V\in\Fl_n(\mathbb{C})$ such that for all $1 \le k \le n-1$, the lexicographically maximal $I\in\binom{[n]}{k}$ such that $\Delta_I(V)\neq 0$ is $I = w([k])$. We have the {\itshape Bruhat decomposition}
$$
\Fl_n(\mathbb{C}) = \bigsqcup_{w\in\mathfrak{S}_n}\mathring{X}^w.
$$
\end{defn}

We now define the general twist map.
\begin{defn}\label{defn_general_twist}
Given $n\in\mathbb{N}$, define the involution $\iota : \GL_n(\mathbb{C}) \to \GL_n(\mathbb{C})$ (called the {\itshape positive inverse}) by
$$
\iota(g) := \delta_n g^{-1}\delta_n.
$$
In other words, $\iota(g)_{i,j} = (-1)^{i+j}(g^{-1})_{i,j}$ for $1 \le i,j \le n$.

Given $V\in\Fl_n(\mathbb{C})$, we define a canonical representative $g\in\U_n$ of $V$ as follows: if $V\in\mathring{X}^w$ (where $w\in\mathfrak{S}_n$), then
\begin{align}\label{canonical_schubert_representative}
\Delta_{w[k]}(g) \in\mathbb{R}_{>0} \quad \text{ and } \quad \Delta_I(g) = 0 \text{ for all } I\in\textstyle\binom{[n]}{k} \text{ with } I >_{\text{lex}} w([k])
\end{align}
for all $1 \le k \le n$. We let $\twist(V)\in\Fl_n(\mathbb{C})$ denote the complete flag represented by $\iota(g)$. This defines the {\itshape (Iwasawa) twist map} $\twist : \Fl_n(\mathbb{C})\to\Fl_n(\mathbb{C})$.
\end{defn}

\begin{rmk}\label{twist_name_remark}
The name {\itshape twist map} is motivated by the twist maps defined by Berenstein, Fomin, and Zelevinsky on $\N_n(\mathbb{C})$ and $\GL_n(\mathbb{C})$ \cite{berenstein_fomin_zelevinsky96,fomin_zelevinsky99}. The key difference is that our map $\twist$ is based on the Iwasawa (or $QR$-) decomposition of $\GL_n(\mathbb{C})$, rather than the Bruhat decomposition. A different twist map was defined on $\Fl_n(\mathbb{C})$ in the latter sense by Galashin and Lam \cite{galashin_lam}.
\end{rmk}

\begin{eg}\label{eg_general_twist}
Let
$$
g := \frac{1}{2}\begin{bmatrix}
\sqrt{2} & -1 & 1 \\[2pt]
\sqrt{2} & 1 & -1 \\[2pt]
0 & \sqrt{2} & \sqrt{2}
\end{bmatrix}\in\U_3, \quad \text{ whence } \iota(g) = \delta_3g^{-1}\delta_3 = \frac{1}{2}\begin{bmatrix}
\sqrt{2} & -\sqrt{2} & 0 \\[2pt]
1 & 1 & -\sqrt{2} \\[2pt]
1 & 1 & \sqrt{2}
\end{bmatrix}.
$$
We can verify that $g$ satisfies \eqref{canonical_schubert_representative} with $w := 231\in\mathfrak{S}_3$. Therefore $\twist:\Fl_3(\mathbb{C})\to\Fl_3(\mathbb{C})$ takes the complete flag represented by $g$ to the complete flag represented by $\iota(g)$. Note that $\iota(g)$ satisfies \eqref{canonical_schubert_representative} for the permutation $w^{-1} = 312$, in agreement with \cref{twist_involution}.
\end{eg}

\begin{prop}\label{twist_involution}
The twist map $\twist$ on $\Fl_n(\mathbb{C})$ is an involution. For each $w\in\mathfrak{S}_n$, it restricts to a diffeomorphism $\mathring{X}^w \to \mathring{X}^{w^{-1}}$.
\end{prop}

\begin{proof}
Let $w\in\mathfrak{S}_n$. The twist map restricted to $\mathring{X}^w$ is smooth, so it suffices to show that given $V\in\mathring{X}^w$, we have $\twist(V)\in\mathring{X}^{w^{-1}}$ and that $\twist(\twist(V)) = V$.

Let $g\in\U_n$ be the canonical representative of $V$ as in \eqref{canonical_schubert_representative}, and write
$$
g = b\mathring{w}b', \quad \text{ where } b,b'\in\B_n(\mathbb{C}).
$$
By the Cauchy--Binet identity \eqref{cauchy-binet} and since $b'$ is upper-triangular, we have
\begin{align}\label{bb_determinant_identity}
\Delta_{w([k])}(g) = \sum_{I,J\in\binom{[n]}{k}}\Delta_{w([k]),I}(b)\Delta_{I,J}(\mathring{w})\Delta_J(b') = \Delta_{w([k]),w([k])}(b)\Delta_{[k]}(b')
\end{align}
for all $0 \le k \le n$. Then by \eqref{canonical_schubert_representative}, we obtain
\begin{align}\label{bb_inequality}
b_{w(k),w(k)}b'_{k,k} = \frac{\Delta_{w([k])}(g)}{\Delta_{w([k-1])}(g)} > 0 \quad \text{ for all } 1 \le k \le n.
\end{align}

By \eqref{permutation_matrix_inverse}, we have
$$
\iota(g) = \delta_n g^{-1}\delta_n = (\delta_nb'^{-1}\delta_n)\mathring{(w^{-1})}(\delta_nb^{-1}\delta_n),
$$
so $\twist(V)\in \mathring{X}^{w^{-1}}$. Also, for all $1 \le k \le n$, applying \eqref{bb_determinant_identity} to $\Delta_{w^{-1}([k])}(\iota(g))$ gives
$$
\Delta_{w^{-1}([k])}(\iota(g)) = \Delta_{w^{-1}([k]),w^{-1}([k])}(\delta_nb'^{-1}\delta_n)\Delta_{[k]}(\delta_nb^{-1}\delta_n) = \prod_{i=1}^k(b'_{w^{-1}(i),w^{-1}(i)}b_{i,i})^{-1} > 0,
$$
using \eqref{bb_inequality}. Hence $\iota(g)$ is the canonical representative of $\twist(V)$ as in \eqref{canonical_schubert_representative}, and since $\iota$ is an involution, we obtain $\twist(\twist(V)) = V$.
\end{proof}

\begin{rmk}\label{schubert_vs_opposite}
We could just as well have defined the twist map $\twist$ using the decomposition of $\Fl_n(\mathbb{C})$ into {\itshape opposite Schubert cells}
$$
\mathring{X}_w := \Bminus_n(\mathbb{C})\cdot\mathring{w}\subseteq\Fl_n(\mathbb{C}) \quad \text{ for } w\in\mathfrak{S}_n,
$$
rather than Schubert cells. The resulting twist map would be different from the one in \cref{defn_general_twist}; the relationship between the two maps can be derived from the fact that $\Bminus_n(\mathbb{C})$ is equal to $\B_n(\mathbb{C})$ conjugated by $\mathring{w_0}$, where $w_0 := (i \mapsto n+1-i)\in\mathfrak{S}_n$. These conventions are ultimately not important for our purposes, because we take absolute values in \cref{extended_moser}.
\end{rmk}

\begin{rmk}\label{general_vs_positive_twist}
In \cref{defn_general_twist}, we have defined the twist map $\twist$ on $\Fl_n(\mathbb{C})$ in a piecewise manner, based on the Bruhat decomposition. While $\twist$ defines a smooth map on each Schubert cell $\mathring{X}^w\subseteq\Fl_n(\mathbb{C})$, in general $\twist$ is not continuous when passing between cells. However, $\twist$ displays remarkable positivity properties, as we explored in \cite{bloch_karp23b}; in particular, it restricts to an involution on the totally nonnegative part $\Fl_n^{\ge 0}$, which extends to a smooth map in an open neighborhood inside $\Fl_n(\mathbb{R})$ \cite[Definition 3.21]{bloch_karp23b}. We emphasize, however, that such a smooth extension differs from the general twist map $\twist$ outside of $\Fl_n^{\ge 0}$, since $\twist$ is not necessarily continuous on the boundary of $\Fl_n^{\ge 0}$. For example, let $g(t)\in\Fl_3(\mathbb{C})$ be represented by the matrix $g$ from \cref{eg_general_twist}, with the $(3,1)$-entry replaced by $t\in\mathbb{R}$. Then $g(0)\in\Fl_3^{\ge 0}$, and $g(t)\notin\Fl_3^{\ge 0}$ for all $t<0$. We have $\lim_{t\to 0}g(t) = g(0)$, but we can verify that
$$
\lim_{t\to 0,\, t<0}\twist(g(t)) = \frac{1}{2}\begin{bmatrix}
-\sqrt{2} & \sqrt{2} & 0 \\[2pt]
-1 & -1 & \sqrt{2} \\[2pt]
1 & 1 & \sqrt{2}
\end{bmatrix}\neq \twist(g(0))\in\Fl_3(\mathbb{C}).
$$
\end{rmk}

Finally, we observe that $\twist$ defines an involution on $\Orbit_{\bflambda}\cong\Fl_n(\mathbb{C})$ under the isomorphism \eqref{orbit_complete_isomorphism_intro}.
\begin{defn}\label{defn_twist_orbit}
We define $\twistorbit : \Orbit_{\bflambda} \to \Orbit_{\bflambda}$ as the involution on $\Orbit_{\bflambda}$ induced by the involution $\twist$ on $\Fl_n(\mathbb{C})$, via the isomorphism \eqref{orbit_complete_isomorphism_intro}. Explicitly,
$$
\twistorbit(g(\ii\Diag{\bflambda})g^{-1}) := \iota(g)(\ii\Diag{\bflambda})(\iota(g))^{-1} = \delta_ng^{-1}(\ii\Diag{\bflambda})g\delta_n
$$
for all $g\in\U_n$ satisfying \eqref{canonical_schubert_representative} (for some $w\in\mathfrak{S}_n$ and all $1 \le k \le n$).
\end{defn}

\section{The symmetric Toda flow as a twisted gradient flow}\label{sec_toda_gradient}

\noindent In this section we use the twist map $\twistorbit$ to show that the symmetric Toda flow is a twisted gradient flow on $\Orbit_{\bflambda}$ in the K\"{a}hler metric (see \cref{full_symmetric_toda_gradient}). This generalizes \cite[Theorem 8.6(ii)]{bloch_karp23b}, where we proved the same result restricted to the totally nonnegative part $\Orbit_{\bflambda}^{\ge 0}$; we refer to \cite[Section 8]{bloch_karp23b} for further discussion and context.

\subsection{Background on the K\"{a}hler metric and gradient flows}\label{sec_gradient_background}
We begin by recalling background on the K\"{a}hler metric on $\Orbit_{\bflambda}$ and gradient flows, following \cite[Section 5]{bloch_karp23b}.
\begin{defn}\label{defn_metric}
Let $\nu$ denote the {\itshape Killing form} on $\gl_n(\mathbb{C})$, given by
$$
\nu(L,M) := 2n\tr(LM) - 2\tr(L)\tr(M) \quad \text{ for all } L,M\in\gl_n(\mathbb{C}).
$$
Then $-\nu(\cdot,\cdot)$ defines a $[\cdot,\cdot]$-invariant pairing (i.e.\ $\nu(\ad_L(M),N) = -\nu(M,\ad_L(N))$) which is positive semidefinite on $\uu_n$.

Now let $L\in\Orbit_{\bflambda}$. For $X\in\uu_n$, we define $X^L$ and $X_L$ by the (unique) decomposition
$$
X = X^L + X_L, \quad \text{ where $X^L\in\im(\ad_L)$ and $X_L\in\ker(\ad_L)$}.
$$
The {\itshape normal metric} on $\Orbit_{\bflambda}$ is given at $L\in\Orbit_{\bflambda}$ by
$$
\langle [L,X],[L,Y]\rangle_{\textnormal{normal}} := -\nu(X^L,Y^L)
$$
for all tangent vectors $[L,X]$ and $[L,Y]$ at $L$. Then the {\itshape K\"{a}hler metric} on $\Orbit_{\bflambda}$ is given at $L\in\Orbit_{\bflambda}$ by
$$
\langle [L,X],[L,Y]\rangle_{\textnormal{K\"{a}hler}} := \langle\textstyle\sqrt{-\hspace*{-2pt}\ad_L^2}([L,X]),[L,Y]\rangle_{\textnormal{normal}},
$$
where $\sqrt{-\hspace*{-2pt}\ad_L^2}$ denote the positive square root of the positive semidefinite operator $-\hspace*{-2pt}\ad_L^2$.
\end{defn}

\begin{defn}\label{defn_gradient_flow}
Given $N\in\uu_n$, we define the {\itshape gradient flow on $\Orbit_{\bflambda}$ with respect to $N$} (in a particular Riemannian metric) as the flow given by
$$
\dot{L}(t) = \grad(H)(L(t)), \quad \text{ where } H(M) := \nu(M,N) \text{ for all } M\in\Orbit_{\bflambda}.
$$
\end{defn}

We have the following explicit description of gradient flows on $\Orbit_{\bflambda}$ in the K\"{a}hler metric:
\begin{prop}[{\cite[Section 3]{duistermaat_kolk_varadarajan83}; \cite[Appendix]{guest_ohnita93}}]\label{gradient_flow_kahler}
Let $L(t)$ evolve according to the gradient flow on $\Orbit_{\bflambda}$ beginning at $L_0$ with respect to $N\in\uu_n$ in the K\"{a}hler metric, and let $V(t)\in\Fl_n(\mathbb{C})$ be the corresponding complete flag under the isomorphism \eqref{orbit_complete_isomorphism_intro}, with $V_0 := V(0)$. Then
\begin{align}\label{gradient_flow_kahler_equation}
V(t) = \exp(t\ii N)V_0 \quad \text{ for all } t\in\mathbb{R}.
\end{align}
Letting $g(t)\in\U_n$ be any representative of $V(t)$, we have $L(t) = g(t)(\ii\Diag{\bflambda})g(t)^{-1}$. Explicitly, we can take $g_0\in\U_n$ representing $V_0$, and then take
\begin{align}\label{gradient_flow_kahler_equation_iwasawa}
g(t) = \Kterm(\exp(t\ii N)g_0) \quad \text{ for all } t\in\mathbb{R}.
\end{align}

\end{prop}

We observe that (as will be useful later) in the case that $N$ is a diagonal matrix, we can explicitly describe the Pl\"{u}cker coordinates of the element $V(t)$ in \eqref{gradient_flow_kahler_equation} in terms of those of $V_0$.
\begin{lem}\label{diagonal_pluckers}
Let $N := -\ii\Diag{c_1, \dots, c_n}\in\uu_n$, and suppose that $V(t)$ evolves according to \eqref{gradient_flow_kahler_equation}. Then
$$
\Delta_I(V(t)) = e^{(\sum_{i\in I}c_i)t}\Delta_I(V_0) \quad \text{ for all } I\subseteq [n].
$$
\end{lem}

\begin{proof}
This follows by a direct calculation, since $\exp(t\ii N) = \Diag{e^{c_1t}, \dots, e^{c_nt}}$.
\end{proof}

\begin{rmk}\label{kahler_flow_decomposition}
We mention that in addition to \eqref{gradient_flow_kahler_equation_iwasawa}, there is another way to obtain an explicit solution to $L(t)$. Namely, let $V(t) = (V_1(t), \dots, V_{n-1}(t))$ be as in \cref{gradient_flow_kahler}. Then as in \eqref{projection_sum_formula}, we write
\begin{align}\label{kahler_flow_decomposition_equation}
-\ii L(t) = \bigg(\sum_{k=1}^{n-1}(\lambda_k - \lambda_{k+1})P_k(t)\bigg) + \lambda_n\I_n,
\end{align}
where $P_k(t)\in\gl_n(\mathbb{C})$ is orthogonal projection onto $V_k(t)$. Regarding elements of $\Gr_{k,n}(\mathbb{C})$ as $n\times k$ matrices (as in \cref{defn_grassmannian}), we have the formula
\begin{multline*}
P_k(t) = V_k(t)(\adjoint{V_k(t)}V_k(t))^{-1}\adjoint{V_k(t)} \\
= \exp(t\ii N)(V_0)_k(\adjoint{(V_0)_k}\exp(2t\ii N)(V_0)_k)^{-1}\adjoint{(V_0)_k}\exp(t\ii N).
\end{multline*}
This leads (via \eqref{kahler_flow_decomposition_equation}) to an explicit expression for $L(t)$  which does not require computing a $QR$-decomposition, as in \eqref{gradient_flow_kahler_equation_iwasawa}.
\end{rmk}

\subsection{The symmetric Toda flow}\label{sec_symmetric_toda}
We now show that the symmetric Toda flow \eqref{toda_flaschka_skew} on $\Orbit_{\bflambda}$ is a twisted gradient flow. This generalizes \cite[Theorem 8.6(ii)]{bloch_karp23b}, where we proved the same result restricted to the totally nonnegative part $\Orbit_{\bflambda}^{\ge 0}$; essentially the same proof applies, using the general twist map $\twistorbit$ defined in \cref{sec_general_twist}.
\begin{thm}\label{full_symmetric_toda_gradient}
Set $N := -\ii\Diag{\bflambda}\in\uu_n$. The symmetric Toda flow \eqref{toda_flaschka_skew} restricted to $\Orbit_{\bflambda}$ is the twisted gradient flow with respect to $N$ in the K\"{a}hler metric. That is, if $L(t)$ evolves according to \eqref{toda_flaschka_skew} beginning at $L_0\in\Orbit_{\bflambda}$, then $\twistorbit(L(t))$ is the gradient flow with respect to $N$ in the K\"{a}hler metric beginning at $\twistorbit(L_0)\in\Orbit_{\bflambda}$ (cf.\ \cref{defn_gradient_flow} and \cref{gradient_flow_kahler}). In the notation of \cref{gradient_flow_kahler}, we have
\begin{align}\label{full_symmetric_toda_gradient_formula}
V(t) = \twist(\Diag{e^{\lambda_1t}, \dots, e^{\lambda_nt}}\cdot\twist(V_0)) \quad \text{ for all } t\in\mathbb{R}.
\end{align}
\end{thm}

\begin{proof}
The proof is the same as in \cite[Proof of Theorem 8.6]{bloch_karp23b}. We only need to observe that if $g_0\in\U_n$ satisfies \eqref{canonical_schubert_representative} for some $w\in\mathfrak{S}_n$ and all $1 \le k \le n$, then so does $\Kterm(\exp(t\ii N)g_0)$.
\end{proof}

\begin{rmk}\label{toda_gradient_normal}
We note that Bloch \cite[Section 6]{bloch90} showed that the tridiagonal Toda lattice flow \eqref{toda_flaschka_skew} (with $L$ tridiagonal) can be written as
$$
\dot{L} = [L, [L,N]], \quad \text{ where } N := -\ii\Diag{n-1, \dots, 1, 0}\in\uu_n,
$$
In particular, by \cite{brockett91,bloch_brockett_ratiu92}, the tridiagonal Toda flow restricted to $\Orbit_{\bflambda}$ is the gradient flow with respect to $N$ in the {\itshape normal} metric. (However, this result does not directly extend to the full symmetric Toda flow; cf.\ \cite{de_mari_pedroni99}.) It is curious that the Toda lattice flow can be written as a gradient flow in two different metrics in rather different ways.
\end{rmk}

\section{Tridiagonal matrices and the Moser map}\label{sec_moser}

\noindent In this section, we explicitly describe Moser's map \eqref{moser_map_intro} and its connection to the tridiagonal Toda lattice flow. We closely follow \cite[Section 4.4]{bloch_karp23b}. In order to use the framework of adjoint orbits, we work with skew-Hermitian matrices $\ii M$ rather than symmetric (or Hermitian) matrices $M$.
\begin{defn}\label{defn_jacobi_matrix}
We define $\Jac_{\bflambda}^{>0}$ (respectively, $\Jac_{\bflambda}^{\ge 0}$) to be the set of elements $\ii M$ of $\Orbit_{\bflambda}$ such that $M$ is a real tridiagonal matrix with positive (respectively, nonnegative) entries immediately above and below the diagonal. Equivalently, $\Jac_{\bflambda}^{>0}$ (respectively, $\Jac_{\bflambda}^{\ge 0}$) is the set of tridiagonal matrices in $\Orbit_{\bflambda}^{>0}$ (respectively, $\Orbit_{\bflambda}^{\ge 0}$) \cite[Proposition 4.18]{bloch_karp23b}. The space $\Jac_{\bflambda}^{>0}$ is known as an {\itshape isospectral manifold of Jacobi matrices}.
\end{defn}

\begin{defn}\label{defn_krylov}
Let $\mathbb{P}^{n-1}_{>0}$ (respectively, $\mathbb{P}^{n-1}_{\ge 0}$) denote the subset of $\mathbb{P}^{n-1}(\mathbb{C})$ where all coordinates are real and positive (respectively, nonnegative), up to rescaling. Given $x\in\mathbb{P}^{n-1}_{>0}$, we define the {\itshape Vandermonde flag} $\Kry{\bflambda}{x}\in\Fl_n(\mathbb{C})$ as the complete flag $(V_1, \dots, V_{n-1})$, where
$$
V_k := \spn(x, \Diag{\bflambda}x, \dots, \Diag{\bflambda}\hspace*{-1pt}^{k-1}x) \quad \text{ for } 1 \le k \le n-1.
$$
Equivalently, $\Kry{\bflambda}{x}$ is represented by the rescaled Vandermonde matrix $(\lambda_i^{j-1}x_i)_{1 \le i,j \le n} \in\GL_n(\mathbb{C})$. We let $\Moser_{\bflambda}(x)\in\Orbit_{\bflambda}$ denote element corresponding to $\Kry{\bflambda}{x}\in\Fl_n(\mathbb{C})$ under the isomorphism \eqref{orbit_complete_isomorphism_intro}. We call $\Moser_{\bflambda}$ the {\itshape Moser map}, since it essentially appeared (with a different, but equivalent, definition) in \cite{moser75}.
\end{defn}

\begin{thm}[{Moser \cite[Section 3]{moser75}; Bloch and Karp \cite[Corollary 4.24]{bloch_karp23b}}]\label{tridiagonal_flag}
The Moser map $\Moser_{\bflambda} : \mathbb{P}^{n-1}_{>0} \xrightarrow{\cong} \Jac_{\bflambda}^{>0}$ is a diffeomorphism.
\end{thm}

\begin{eg}\label{eg_tridiagonal_flag}
Let $\bflambda := (1,0,-1)$. Then the Moser map $\Moser_{\bflambda}$ sends $(x_1 : x_2 : x_3)\in\mathbb{P}^2_{>0}$ to
\begin{gather*}
\ii\begin{bmatrix}
\frac{x_1^2 - x_3^2}{x_1^2 + x_2^2 + x_3^2} & \frac{\sqrt{x_1^2x_2^2 + 4x_1^2x_3^2 + x_2^2x_3^2}}{x_1^2 + x_2^2 + x_3^2} & 0 \\[8pt]
\frac{\sqrt{x_1^2x_2^2 + 4x_1^2x_3^2 + x_2^2x_3^2}}{x_1^2 + x_2^2 + x_3^2} & \frac{(x_1^2 - x_3^2)(x_2^4 - 4x_1^2x_3^2)}{(x_1^2 + x_2^2 + x_3^2)(x_1^2x_2^2 + 4x_1^2x_3^2 + x_2^2x_3^2)} & \frac{2x_1x_2x_3\sqrt{x_1^2 + x_2^2 + x_3^2}}{x_1^2x_2^2 + 4x_1^2x_3^2 + x_2^2x_3^2} \\[8pt]
0 & \frac{2x_1x_2x_3\sqrt{x_1^2 + x_2^2 + x_3^2}}{x_1^2x_2^2 + 4x_1^2x_3^2 + x_2^2x_3^2} & \frac{x_2^2(x_3^2 - x_1^2)}{x_1^2x_2^2 + 4x_1^2x_3^2 + x_2^2x_3^2} \\[4pt]
\end{bmatrix}\in\Jac_{\bflambda}^{>0}.\qedhere
\end{gather*}

\end{eg}

We now discuss the topology of the compact isospectral manifold $\Jac_{\bflambda}^{\ge 0}$.
\begin{defn}\label{defn_permutohedron}
Let $\Perm{\bflambda}\subseteq\mathbb{R}^n$ denote the polytope whose vertices are all $n!$ permutations of $\bflambda = (\lambda_1, \dots, \lambda_n)$. We call $\Perm{\bflambda}$ a {\itshape permutohedron}. We also define the {\itshape moment map}
$$
\mu : \uu_n \to \mathbb{R}^n, \quad \ii M \mapsto \diag(M).
$$
\end{defn}

\begin{eg}\label{eg_permutohedron}
Let $\bflambda := (1,0,-1)$. Then $\Perm{\bflambda}$ is a hexagon in $\mathbb{R}^3$, contained in the hyperplane where all coordinates sum to zero:
\begin{gather*}
\quad\begin{tikzpicture}[baseline=(current bounding box.center)]
\tikzstyle{out1}=[inner sep=0,minimum size=1.2mm,circle,draw=black,fill=black,semithick]
\tikzstyle{vertex}=[inner sep=0,minimum size=1.2mm,circle,draw=red,fill=red,semithick]
\pgfmathsetmacro{\r}{1.20};
\pgfmathsetmacro{\s}{1.0};
\pgfmathsetmacro{\hs}{1.00};
\pgfmathsetmacro{\vs}{0.36};
\draw[thick,fill=black!18](-90:\r)--(-30:\r)--(30:\r)--(90:\r)--(150:\r)--(210:\r)--cycle;
\node[out1](123)at(-90:\r){};
\node[out1](132)at(-30:\r){};
\node[out1](231)at(30:\r){};
\node[out1](321)at(90:\r){};
\node[out1](312)at(150:\r){};
\node[out1](213)at(210:\r){};
\node[inner sep=0]at($(123)+(0,-\vs)$){\scalebox{\s}{$(-1,0,1)$}};
\node[inner sep=0]at($(132)+(\hs,0)$){\scalebox{\s}{$(-1,1,0)$}};
\node[inner sep=0]at($(231)+(\hs,0)$){\scalebox{\s}{$(0,1,-1)$}};
\node[inner sep=0]at($(321)+(0,\vs)$){\scalebox{\s}{$(1,0,-1)$}};
\node[inner sep=0]at($(312)+(-\hs,0)$){\scalebox{\s}{$(1,-1,0)$}};
\node[inner sep=0]at($(213)+(-\hs,0)$){\scalebox{\s}{$(0,-1,1)$}};
\end{tikzpicture}\quad.\qedhere
\end{gather*}

\end{eg}

Tomei \cite[Section 4]{tomei84} showed that $\Jac_{\bflambda}^{\ge 0}$ is homeomorphic to $\Perm{\bflambda}$ (in fact, they are isomorphic as regular CW complexes). We will need the following explicit version of this result due to Bloch, Flashcka, and Ratiu \cite{bloch_flaschka_ratiu90}. It may be regarded as an analogue of the {\itshape Schur--Horn theorem} \cite{schur23,horn54}, which states that $\mu$ maps $\Orbit_{\bflambda}$ surjectively onto $\Perm{\bflambda}$.
\begin{thm}[{Bloch, Flaschka, and Ratiu \cite[Theorem p.\ 60]{bloch_flaschka_ratiu90}; cf.\ \cite[Remark 8.8]{bloch_karp23b}}]\label{jacobi_manifold}
The map
$$
\Jac_{\bflambda}^{\ge 0} \to \Perm{\bflambda}, \quad L \mapsto \mu(\twistorbit(L))
$$
is a homeomorphism which restricts to a diffeomorphism from $\Jac_{\bflambda}^{>0}$ to the interior of $\Perm{\bflambda}$.
\end{thm}

We emphasize that in general, the map
$$
\Jac_{\bflambda}^{\ge 0} \to \Perm{\bflambda}, \quad L \mapsto \mu(L)
$$
(not involving the twist map) is neither injective nor surjective.

\section{Tridiagonalization of the symmetric Toda flow}\label{sec_tridiagonalization}

\subsection{The construction of Deift, Li, Nanda, and Tomei}\label{sec_DLNT_review}
We recall the tridiagonalization of the symmetric Toda flow constructed by Deift, Li, Nanda, and Tomei \cite[Section 7]{deift_li_nanda_tomei86}. While we maintain the notation of \cite{deift_li_nanda_tomei86} where possible, we give an exposition tailored to our perspective. One key notational difference is that we reverse the order of the ground set $[n]$ from its use in \cite{deift_li_nanda_tomei86} (e.g.\ in \cref{defn_hermitian_tridiagonalization}, we use $e_1$ rather than $e_n$); our convention is consistent with the rest of this paper and with \cite{bloch_karp23b,moser75}.

\begin{defn}\label{defn_cyclic_subspace}
Let $A\in\gl_n(\mathbb{C})$ and $x\in\mathbb{C}^n$. We define the {\itshape cyclic subspace}
$$
\cyc{A}{x} := \spn(A^jx : j\in\mathbb{N}) \subseteq \mathbb{C}^n.
$$
Equivalently, $\cyc{A}{x}$ is the minimal $A$-invariant subspace of $\mathbb{C}^n$ containing $x$.
\end{defn}

\begin{defn}\label{defn_hermitian_tridiagonalization}
Let $M$ be an $n\times n$ Hermitian matrix. Let $m := \dim(\cyc{M}{e_1})$, and let $A$ be the $n\times m$ matrix with orthonormal columns obtained by applying the Gram--Schmidt orthonormalization process to the matrix
$$
\begin{bmatrix}
~\\[-22pt]
\vline height 2pt & \vline height 2pt & & \vline height 2pt \\[1pt]
e_1 & Me_1 & \cdots & M^{m-1}e_1 \\[-9pt]
\vline height 2pt & \vline height 2pt & & \vline height 2pt \\
\end{bmatrix}
$$
(so in particular, $\adjoint{A}\hspace*{-1pt}A = \I_m$). We define the $m\times m$ Hermitian matrix $M_T := \adjoint{A}MA$, which represents the endomorphism $M$ restricted to $\cyc{M}{e_1}$. In light of \cref{hermitian_tridiagonalization_properties}\ref{hermitian_tridiagonalization_properties_tridiagonal}, we call $M_T$ the {\itshape tridiagonalization} of $M$.
\end{defn}

\begin{prop}\label{hermitian_tridiagonalization_properties}
Let $M$ be an $n\times n$ Hermitian matrix, and write $M = g\Diag{\bfnu}g^{-1}$, where $\nu_1 \ge \cdots \ge \nu_n$ and $g\in\U_n$.
\begin{enumerate}[label=(\roman*), leftmargin=*, itemsep=2pt]
\item\label{hermitian_tridiagonalization_properties_tridiagonal} The matrix $M_T$ is tridiagonal with positive real entries immediately above and below the diagonal.
\item\label{hermitian_tridiagonalization_properties_eigenvalues} The eigenvalues of $M_T$ are distinct, and its set of eigenvalues is
$$
\{\nu_i : 1 \le i \le n \text{ and } g_{1,i} \neq 0\}.
$$
\item\label{hermitian_tridiagonalization_properties_eigenvectors} Let $\nu_i$ (where $1 \le i \le n$) be a simple eigenvalue of $M$ which is also an eigenvalue of $M_T$. Then the corresponding normalized eigenvectors of $M$ and $M_T$ have the same first entry. That is, if $x\in\mathbb{C}^m$ has norm $1$ such that $M_Tx = \nu_ix$, then $|x_1| = |g_{1,i}|$.
\end{enumerate}

\end{prop}

\begin{proof}
\ref{hermitian_tridiagonalization_properties_tridiagonal} Let $m$ and $A$ be as in \cref{defn_hermitian_tridiagonalization}. We show, equivalently, that the $m\times m$ matrix
$$
\begin{bmatrix}
~\\[-22pt]
\vline height 2pt & \vline height 2pt & & \vline height 2pt \\[1pt]
e_1 & M_Te_1 & \cdots & (M_T)^{m-1}e_1 \\[-9pt]
\vline height 2pt & \vline height 2pt & & \vline height 2pt \\
\end{bmatrix}
$$
is upper-triangular with positive diagonal entries. To this end, let $1 \le j \le m$. By the definition of $A$, we can write
$$
M^{j-1}e_1 = \sum_{i=1}^j c_iAe_i, \quad \text{ where } c_1, \dots, c_j\in\mathbb{C} \text{ and } c_j > 0.
$$
Recall that $M_T$ represents the endomorphism $M$ restricted to $\cyc{M}{e_1}$; similarly, $(M_T)^{j-1} = \adjoint{A}M^{j-1}A$ represents $M^{j-1}$ restricted to $\cyc{M}{e_1}$. We obtain
$$
(M_T)^{j-1}e_1 = \adjoint{A}M^{j-1}Ae_1 = \adjoint{A}M^{j-1}e_1 = \sum_{i=1}^jc_i\adjoint{A}\hspace*{-1pt}Ae_i = \sum_{i=1}^jc_ie_i.
$$

\ref{hermitian_tridiagonalization_properties_eigenvalues} The matrix $M_T$ represents the endomorphism $M$ restricted to $\cyc{M}{e_1}$, which in turn (conjugating by $g^{-1}$) is similar to the endomorphism $\Diag{\bfnu}$ restricted to $\cyc{\Diag{\bfnu}}{g^{-1}e_1}$. Given $1 \le i \le n$, consider the vector $y\in\mathbb{C}^n$ defined by
$$
y_j :=  \begin{cases}
(g^{-1}e_1)_j = \overline{g_{1,j}}, & \text{ if $\nu_j = \nu_i$}; \\
0, & \text{ otherwise}
\end{cases} \quad \text{ for } 1 \le j \le n.
$$
Then by Vandermonde's identity, the distinct nonzero $y$'s form a basis for $\cyc{\Diag{\bfnu}}{g^{-1}e_1}$.

\ref{hermitian_tridiagonalization_properties_eigenvectors} Since $\nu_i$ is a simple eigenvalue of both $M$ and $M_T$, the corresponding eigenvector $ge_i$ of $M$ lies in $\cyc{M}{e_1}$. Therefore $ge_i = Ax$ for some $x\in\mathbb{C}^m$, and $\|x\| = \|e_i\| = 1$ because $g$ and $A$ both have orthonormal columns. We have
$$
M_Tx = \adjoint{A}MAx = \adjoint{A}Mge_i = \nu_i\adjoint{A}ge_i = \nu_i\adjoint{A}\hspace*{-1pt}Ax = \nu_ix,
$$
and $|x_1| = |(\adjoint{A}ge_i)_1| = |(ge_i)_1| = |g_{1,i}|$.
\end{proof}

\begin{eg}\label{eg_hermitian_tridiagonalization}
Let
$$
M := \frac{1}{4}\begin{bmatrix}
-2 & 3\sqrt{2} & 3\sqrt{2} \\[2pt]
3\sqrt{2} & -1 & -1 \\[2pt]
3\sqrt{2} & -1 & -1
\end{bmatrix} = g\Diag{1,0,-2}g^{-1}, \quad \text{ where } g := \frac{1}{2}\begin{bmatrix}
\sqrt{2} & 0 & \sqrt{2} \\[1pt]
1 & -\sqrt{2} & -1 \\[1pt]
1 & \sqrt{2} & -1
\end{bmatrix}\in\U_3.
$$
Since $M^2e_1 = 2e_1 - Me_1$, the cyclic subspace $\cyc{M}{e_1}$ is $2$-dimensional. Taking the matrix
$$
\begin{bmatrix}
~\\[-22pt]
\vline height 2pt & \vline height 2pt \\[1pt]
e_1 & Me_1 \\[-9pt]
\vline height 2pt & \vline height 2pt \\
\end{bmatrix} = \begin{bmatrix}
1 & -\frac{1}{2} \\[4pt]
0 & \frac{3}{2\sqrt{2}} \\[6pt]
0 & \frac{3}{2\sqrt{2}}
\end{bmatrix}
$$
and applying the Gram--Schmidt orthonormalization process gives $A := \begin{bmatrix}
1 & 0 \\[2pt]
0 & \frac{1}{\sqrt{2}} \\[6pt]
0 & \frac{1}{\sqrt{2}}
\end{bmatrix}$. Then
$$
M_T = \adjoint{A}MA = \frac{1}{2}\begin{bmatrix}
-1 & 3 \\
3 & -1
\end{bmatrix} = h\Diag{1,-2}h^{-1}, \quad \text{ where } h := \frac{1}{\sqrt{2}}\begin{bmatrix}
1 & -1 \\
1 & 1
\end{bmatrix}\in\U_2.
$$
We can verify that the properties of $M_T$ in \cref{hermitian_tridiagonalization_properties} hold in this case.
\end{eg}

We now consider actions of $\GL_n(\mathbb{C})$ and $\gl_n(\mathbb{C})$ on the exterior power $\bigwedge^{\hspace*{-1pt}k}(\mathbb{C}^n)$.
\begin{defn}\label{defn_wedge_matrix}
For $0 \le k \le n$, let $\bigwedge^{\hspace*{-1pt}k}(\mathbb{C}^n)$ denote the $k$th exterior power of $\mathbb{C}^n$, which we identify with $\mathbb{C}^{\binom{n}{k}}$ by fixing the basis
\begin{align}\label{wedge_basis}
\{e_{i_1} \wedge \cdots \wedge e_{i_k} : 1 \le i_1 < \cdots < i_k \le n\}
\end{align}
ordered lexicographically. Given $g\in\GL_n(\mathbb{C})$, we let $g^{(k)}\in\GL(\bigwedge^{\hspace*{-1pt}k}(\mathbb{C}^n)) \cong \GL_{\binom{n}{k}}(\mathbb{C})$ denote the map induced by $g$, i.e.,
$$
g^{(k)}(x_1 \wedge \cdots \wedge x_k) := gx_1 \wedge \cdots \wedge gx_k \quad \text{ for all } x_1, \dots, x_k\in\mathbb{C}^n.
$$
Given $M\in\gl_n(\mathbb{C})$, we let $M_{(k)}\in\gl(\bigwedge^{\hspace*{-1pt}k}(\mathbb{C}^n)) \cong \gl_{\binom{n}{k}}(\mathbb{C})$ denote the vector field induced by $M$, i.e.,
$$
M_{(k)}(x_1 \wedge \cdots \wedge x_k) := \sum_{i=1}^kx_1 \wedge \cdots \wedge x_{i-1} \wedge Mx_i \wedge x_{i+1} \wedge \cdots \wedge x_k \quad \text{ for all } x_1, \dots, x_k\in\mathbb{C}^n.
$$
Equivalently, $M_{(k)} = \frac{d}{dt}\eval{t=0}\exp(tM)^{(k)}$.
\end{defn}

Deift, Li, Nanda, and Tomei \cite{deift_li_nanda_tomei86} observe that the operations $M\mapsto M_T$ and $M\mapsto M_{(k)}$ are both compatible with the symmetric Toda flow (the proofs in \cite{deift_li_nanda_tomei86} are over $\mathbb{R}$, but easily extend over $\mathbb{C}$). In the statements below, given a Hermitian matrix $N$, we let $N(t)$ denote the symmetric Toda flow \eqref{toda_flaschka_symmetric} beginning at $N$.
\begin{prop}[{Deift, Li, Nanda, and Tomei \cite[Propositions 7.2 and 7.3]{deift_li_nanda_tomei86}}] \label{DLNT_toda_compatible}
Let $M$ be an $n\times n$ Hermitian matrix, let $0 \le k \le n$, and let $t\in\mathbb{R}$. Then the following two diagrams commute:
$$
\begin{tikzcd}[row sep=40pt, column sep=60pt]
M \arrow[mapsto]{r}{\textnormal{Toda}} \arrow[mapsto]{d}{\cdot_T} & M(t) \arrow[mapsto]{d}{\cdot_T} \\
M_T \arrow[mapsto]{r}{\textnormal{Toda}} & M_T(t) = M(t)_T
\end{tikzcd} \quad \text{ and } \quad \begin{tikzcd}[row sep=40pt, column sep=60pt]
M \arrow[mapsto]{r}{\textnormal{Toda}} \arrow[mapsto]{d}{\cdot_{(k)}} & M(t) \arrow[mapsto]{d}{\cdot_{(k)}} \\
M_{(k)} \arrow[mapsto]{r}{\textnormal{Toda}} & M_{(k)}(t) = M(t)_{(k)}
\end{tikzcd}.
$$
\end{prop}

Deift, Li, Nanda, and Tomei \cite{deift_li_nanda_tomei86} then apply \cref{DLNT_toda_compatible} to embed the symmetric Toda flow into a product of tridiagonal Toda flows:
\begin{thm}[{Deift, Li, Nanda, and Tomei \cite[Theorem p.\ 230]{deift_li_nanda_tomei86}}] \label{DLNT_tridiagonalization}
Let $M$ be an $n\times n$ Hermitian matrix such that for all $1 \le k \le n-1$, the $\binom{n}{k}$ sums of $k$ distinct eigenvalues of $M$ are pairwise distinct. Then for all $t\in\mathbb{R}$, the following diagram commutes:
\begin{equation}\label{DLNT_tridiagonalization_commutative_diagram}
\begin{tikzcd}[row sep=40pt, column sep=80pt]
M \arrow[mapsto]{r}{\textnormal{Toda}} \arrow[mapsto]{d}{\big((\cdot_{(k)})_T\big)_{1 \le k \le n-1}} & M(t) \arrow[mapsto]{d}{\big((\cdot_{(k)})_T\big)_{1 \le k \le n-1}} \\
\big((M_{(k)})_T\big)_{1 \le k \le n-1} \arrow[mapsto]{r}{\textnormal{Toda}} & \big((M_{(k)})_T(t)\big)_{1 \le k \le n-1} = \big((M(t)_{(k)})_T\big)_{1 \le k \le n-1}
\end{tikzcd}.
\end{equation}
Moreover, the map
\begin{align}\label{DLNT_tridiagonalization_embedding}
t \mapsto \big((M_{(k)})_T(t)\big)_{1 \le k \le n-1} = \big((M(t)_{(k)})_T\big)_{1 \le k \le n-1} \quad \text{ for } t\in\mathbb{R}
\end{align}
is injective. That is, a generic trajectory of the symmetric Toda flow embeds into a product of trajectories of tridiagonal Toda flows.
\end{thm}

\subsection{Tridiagonalization via the twist map, Pl\"{u}cker embedding, and Moser map}\label{sec_DLNT_review}
We now rephrase \cref{DLNT_tridiagonalization} (and give a new proof) using three maps we introduced earlier: the twist map $\twist$, the Pl\"{u}cker embedding $\Delta$, and the Moser map $\Moser_{\bflambda}$. For convenience, we extend the domain of $\Moser_{\bflambda}$ from $\mathbb{P}^{n-1}_{>0}$ to $\mathbb{P}^{n-1}(\mathbb{C})$:
\begin{defn}\label{extended_moser}
Recall the Moser map $\Moser_{\bflambda} : \mathbb{P}^{n-1}_{>0} \to \Orbit_{\bflambda}$ defined in \cref{defn_krylov}. We define the {\itshape extended Moser map}
$$
\Moserextend_{\bflambda} : \mathbb{P}^{n-1}(\mathbb{C}) \to \bigcup_{\bfnu}\Orbit_{\bfnu}
$$
as follows, where the union above is over all nonempty subsequences $\bfnu$ of $\bflambda$. Given $x\in\mathbb{P}^{n-1}(\mathbb{C})$, we set $I := \{i \in [n] : x_i\neq 0\}$, and define
$$
y := (|x_i| : i\in I)\in\mathbb{P}^{|I|-1}_{>0} \quad \text{ and } \quad \bfnu := (\lambda_i : i\in I).
$$
Then we define $\Moserextend_{\bflambda}(x) := \Moser_{\bfnu}(y)$.
\end{defn}

\begin{eg}\label{eg_extended_moser}
We have $\Moserextend_{(\lambda_1, \lambda_2, \lambda_3)}(-1 : 0 : 3i + 4) = \Moser_{(\lambda_1, \lambda_3)}(1 : 5)$.
\end{eg}

We can write the tridiagonal Toda lattice flow in terms of the Moser map, as follows.
\begin{prop}[{Moser \cite[(1.4)]{moser75}}]\label{toda_in_moser}
Let $L\in\Jac_{\bflambda}^{>0}$, and let $L(t)$ denote the symmetric Toda flow \eqref{toda_flaschka_skew} beginning at $L$. Take $x\in\mathbb{P}^{n-1}_{>0}$ as in \cref{tridiagonal_flag} such that $\Moser_{\bflambda}(x) = L$. Then
$$
L(t) = \Moser_{\bflambda}(e^{\lambda_1t}x_1 : \cdots : e^{\lambda_nt}x_n) \quad \text{ for all } t\in\mathbb{R}.
$$
\end{prop}

\begin{proof}
This is a restatement of \eqref{moser_gradient_intro}. Alternatively, we can apply \cref{full_symmetric_toda_gradient} and \eqref{gradient_flow_kahler_equation}.
\end{proof}

We recall that we are working with both Hermitian and skew-Hermitian matrices; if $M$ is an $n\times n$ Hermitian matrix with distinct eigenvalues $\bflambda$, then we have a corresponding skew-Hermitian matrix $L = \ii M\in\Orbit_{\bflambda}$.
\begin{prop}\label{tridiagonalization_from_moser}
Let $M$ be an $n\times n$ Hermitian matrix with distinct eigenvalues $\bflambda$, and write $M = g\Diag{\bflambda}g^{-1}$, where $g\in\U_n$. Then
$$
\Moserextend_{\bflambda}(g_{1,1} : \cdots : g_{1,n}) = \ii M_T.
$$
\end{prop}

\begin{proof}
This follows from \cref{hermitian_tridiagonalization_properties} and \cref{tridiagonal_flag}.
\end{proof}

We now generalize \cref{tridiagonalization_from_moser} to give an interpretation of $(M_{(k)})_T$. For $0 \le k \le n$, we define
$$
\bflambda^{(k)} := \bigg(\sum_{i\in I}\lambda_i : I\in\textstyle\binom{[n]}{k}\bigg)\in\mathbb{R}^{\binom{n}{k}}.
$$
Moreover, we reorder the elements of $\binom{[n]}{k}$ (from the lexicographic order) so that the entries of $\bflambda^{(k)}$ are weakly decreasing. Importantly, in both orders, the first element of $\binom{[n]}{k}$ is $[k]$.
\begin{cor}\label{wedge_matrix_from_moser}
Let $M$ be an $n\times n$ Hermitian matrix with distinct eigenvalues $\bflambda$, and let $V\in\Fl_n(\mathbb{C})$ denote the complete flag corresponding to $\ii M$ under the isomorphism \eqref{orbit_complete_isomorphism_intro}. Let $0 \le k \le n$ be such that the entries of $\bflambda^{(k)}$ are distinct. Then
$$
\Moserextend_{\bflambda^{(k)}}\big(\Delta_I(\twist(V)) : I\in\textstyle\binom{[n]}{k}\big) = \ii(M_{(k)})_T.
$$
\end{cor}

\begin{proof}
Let $g\in\U_n$ be the canonical representative of $V$ as in \eqref{canonical_schubert_representative}, so that $M = g\Diag{\bflambda}g^{-1}$. Then the eigenvalues of $M$ are $\bflambda^{(k)}_I$ for $I\in\binom{[n]}{k}$, with corresponding eigenvectors $\bigwedge_{i\in I}ge_i$. The first entry of $\bigwedge_{i\in I}ge_i$ is the coefficient of $e_1 \wedge \cdots \wedge e_k$, namely, $\Delta_{[k],I}(g)$. We have
$$
|\Delta_{[k],I}(g)| = |\Delta_I(\iota(g))| = |\Delta_I(\twist(V))|,
$$
so the result follows from \cref{tridiagonalization_from_moser}.
\end{proof}

We now state our main result:
\begin{thm}\label{tridiagonalization_via_plucker}
Let $M$ be a Hermitian matrix with distinct eigenvalues $\bflambda$ such that for all $1 \le k \le n-1$, the entries of $\lambda^{(k)}$ are distinct. Let $V\in\Fl_n(\mathbb{C})$ denote the complete flag corresponding to $\ii M$ under the isomorphism \eqref{orbit_complete_isomorphism_intro}. Then under the Toda flow, we have
\begin{align}\label{tridiagonalization_via_plucker_formula}
(M_{(k)})_T(t) = (M(t)_{(k)})_T = -\ii\Moserextend_{\bflambda^{(k)}}\big(e^{\lambda^{(k)}_It}\Delta_I(\twist(V)) : I\in\textstyle\binom{[n]}{k}\big)
\end{align}
for all $t\in\mathbb{R}$ and $1 \le k \le n-1$, which we see explicitly is a twisted gradient flow. In particular, the diagram \eqref{DLNT_tridiagonalization_commutative_diagram} commutes, and we can write the embedding \eqref{DLNT_tridiagonalization_embedding} as
\begin{align}\label{tridiagonalization_via_plucker_embedding}
t \mapsto \Big(-\ii\Moserextend_{\bflambda^{(k)}}\big(e^{\lambda^{(k)}_It}\Delta_I(\twist(V)) : I\in\textstyle\binom{[n]}{k}\big)\Big)_{1 \le k \le n-1} \quad \text{ for } t\in\mathbb{R}.
\end{align}
This embeds a generic trajectory of the symmetric Toda flow into a product of trajectories of tridiagonal Toda flows.
\end{thm}

\begin{proof}
The fact that $(M_{(k)})_T(t)$ equals the right-hand side of \eqref{tridiagonalization_via_plucker_formula} follows from \cref{wedge_matrix_from_moser} and \cref{toda_in_moser}. The fact that $(M(t)_{(k)})_T$ equals the right-hand side of \eqref{tridiagonalization_via_plucker_formula} follows similarly, where instead of \cref{toda_in_moser}, we use \eqref{full_symmetric_toda_gradient_formula} and \cref{diagonal_pluckers}. The fact that \eqref{tridiagonalization_via_plucker_embedding} is an embedding follows from \cref{tridiagonal_flag}.
\end{proof}

\begin{rmk}\label{embedding_remark}
Suppose $\bflambda$ is such that for all $1 \le k \le n-1$, the entries of $\lambda^{(k)}$ are distinct. While \eqref{DLNT_tridiagonalization_embedding} (equivalently, \eqref{tridiagonalization_via_plucker_embedding}) is an embedding, the map
\begin{align}\label{embedding_remark_formula}
M \mapsto \big((M_{(k)})_T\big)_{1 \le k \le n-1} = \Big(-\ii\Moserextend_{\bflambda^{(k)}}\big(\Delta_I(\twist(V)) : I\in\textstyle\binom{[n]}{k}\big)\Big)_{1 \le k \le n-1}
\end{align}
(where $V\in\Fl_n(\mathbb{C})$ denotes the complete flag corresponding to $\ii M$ under the isomorphism \eqref{orbit_complete_isomorphism_intro}) is not injective on the set of $n\times n$ Hermitian matrices $M$ with eigenvalues $\bflambda$. Indeed, we can see from the expression on the right-hand side above that for each $1 \le k \le n-1$, we only recover the vector $\big(|\Delta_I(\twist(V))| : I\in\textstyle\binom{[n]}{k}\big)$ of absolute values of the Pl\"{u}cker coordinates of $\twist(V)_k$, and not necessarily the complex argument of each Pl\"{u}cker coordinate. If in addition we record the complex argument of each $\Delta_I(\twist(V))$ (which, up to a change in notation, is the element $\tau_I(M)$ of \cite[Section 7]{deift_li_nanda_tomei86}), then we obtain an injection. For example, the map \eqref{embedding_remark_formula} gives an embedding when restricted to $\ii M\in\Orbit_{\bflambda}^{>0}$, since in this case every Pl\"{u}cker coordinate $\Delta_I(\twist(V))$ is positive.
\end{rmk}

\begin{rmk}\label{integrability_remark}
Deift, Li, Nanda, and Tomei \cite[Theorem 3.3]{deift_li_nanda_tomei86} also show that the symmetric Toda flow is Liouville integrable, which depends on finding sufficient independent integrals in involution with respect to the appropriate Poisson bracket. As they discuss, the explicit embedding constructed in this section gives an alternative demonstration of integrability, which is independent of Liouville integrability. It would be interesting to further study the connection between these two approaches.
\end{rmk}

\section{Toda flows on moment polytopes}\label{sec_moment_polytopes}\label{moment_sec}

\noindent We recall the discussion of the moment map $\mu$ from \cref{sec_moser}, and in particular \cref{jacobi_manifold}, which implies that the twisted moment map $\mu\circ\twistorbit$ gives a diffeomorphism from $\Jac_{\bflambda}^{>0}$ to the interior of $\Perm{\bflambda}$. In this section, we study variations of this map applied not just to tridiagonal matrices $\Jac_{\bflambda}^{\ge 0}$, but to a general adjoint orbit $\Orbit_{\bflambda}$. One such variation will be based on the embedding of \cref{sec_tridiagonalization}.

We begin by studying the twisted moment map $\mu\circ\twistorbit$ on $\Orbit_{\bflambda}$. The following result is closely related to work of Kodama and Williams \cite[Section 6]{kodama_williams15}.
\begin{prop}\label{twisted_moment_embedding}
The twisted moment map
\begin{align}\label{twisted_moment_embedding_formula}
\Orbit_{\bflambda}\to\Perm{\bflambda}, \quad L \mapsto \mu(\twistorbit(L))
\end{align}
is an embedding when restricted to a generic trajectory of the symmetric Toda flow \eqref{toda_flaschka_skew}.
\end{prop}

\begin{proof}
Given $L\in\Orbit_{\bflambda}$, let $L(t)$ ($t\in\mathbb{R}$) denote the symmetric Toda flow beginning at $L$, and let $V(t)\in\Fl_n(\mathbb{C})$ denote the complete flag corresponding to $L(t)$ under the isomorphism \eqref{orbit_complete_isomorphism_intro}. Then by \eqref{full_symmetric_toda_gradient_formula}, we have
\begin{align}\label{twist_orbit_to_twist_formula}
\twist(V(t)) = \Diag{e^{\lambda_1t}, \dots, e^{\lambda_nt}}\cdot\twist(V_0).
\end{align}
Now let $\H_n(\mathbb{C})$ denote the subset of diagonal matrices in $\GL_n(\mathbb{C})$, and let $\H_n^{>0}$ denote its positive part, consisting of diagonal matrices with positive diagonal entries. Then $\H_n(\mathbb{C})$ acts on $\Fl_n(\mathbb{C})$ by left multiplication. We may regard $\mu$ as a moment map on $\Fl_n(\mathbb{C})$, and therefore it maps a generic $\H_n^{>0}$-orbit homeomorphically onto the interior of the moment polytope $\Perm{\bflambda}$ (see, e.g., \cite[Section 4.2]{fulton93}). Since $\Diag{e^{\lambda_1t}, \dots, e^{\lambda_nt}}\in\H_n^{>0}$ for all $t\in\mathbb{R}$, the result follows.
\end{proof}

\begin{rmk}\label{moment_map_explicit_remark}
We note that one can write down the moment map $\mu$ explicitly in terms of Pl\"{u}cker coordinates, as follows. Let $\hypersimplex_{k,n}$ denote the convex hull in $\mathbb{R}^n$ of all $\binom{n}{k}$ vectors with $k$ $1$'s and $n-k$ $0$'s, called a {\itshape hypersimplex}. Then we have the {\itshape Grassmannian moment map}
$$
\mu_{k,n} : \Gr_{k,n}(\mathbb{C})\to\mathbb{R}^n, \quad V\mapsto \Bigg(\frac{\sum_{I\in\binom{[n]}{k},\hspace*{2pt} I\ni i}|\Delta_I(V)|^2}{\sum_{I\in\binom{[n]}{k}}|\Delta_I(V)|^2}\Bigg)_{1 \le i \le n},
$$
whose image is $\hypersimplex_{k,n}$ (cf.\ \cite[Section 2]{gelfand_goresky_macpherson_serganova87}). One can verify that if $P\in\gl_n(\mathbb{C})$ is orthogonal projection onto $V\in\Gr_{k,n}(\mathbb{C})$, then $\diag(P) = \mu_{k,n}(V)$.

Now let $L\in\Orbit_{\bflambda}$, and let $V = (V_1, \dots, V_{n-1}) \in\Fl_n(\mathbb{C})$ denote the complete flag corresponding to $L$ under the isomorphism \eqref{orbit_complete_isomorphism_intro}. Recall from \eqref{projection_sum_formula} that we may write
$$
-\ii L = \bigg(\sum_{k=1}^{n-1}(\lambda_k - \lambda_{k+1})P_k\bigg) + \lambda_n\I_n,
$$
where $P_k\in\gl_n(\mathbb{C})$ is orthogonal projection onto $V_k$. Then
\begin{gather}\label{moment_map_explicit_formula}
\begin{aligned}
\mu(L) &= \bigg(\sum_{k=1}^{n-1}(\lambda_k - \lambda_{k+1})\diag(P_k)\bigg) + \lambda_n\diag(I_n) \\
&= \bigg(\sum_{k=1}^{n-1}(\lambda_k - \lambda_{k+1})\mu_{k,n}(V_k)\bigg) + \lambda_n(1, \dots, 1).
\end{aligned}
\end{gather}

We note that \eqref{moment_map_explicit_formula} provides a convenient way to calculate \eqref{twisted_moment_embedding_formula} for a trajectory $L(t)$ of the Toda flow, using the formula \eqref{twist_orbit_to_twist_formula}. Namely, we have
$$
\mu(\twistorbit(L(t))) = \bigg(\sum_{k=1}^{n-1}(\lambda_k - \lambda_{k+1})\mu_{k,n}\big(\Diag{e^{\lambda_1t}, \dots, e^{\lambda_nt}}\cdot\twist(V_k)\big)\bigg) + \lambda_n(1, \dots, 1).
$$
\end{rmk}

We now apply the moment map $\mu$ to the image of the embedding used in \cref{tridiagonalization_via_plucker}.
\begin{prop}\label{moment_embedding_permutohedron_product}
Suppose $\bflambda\in\mathbb{R}^n$ is such that for all $1 \le k \le n-1$, the entries of $\lambda^{(k)}$ are distinct. Given $L\in\Orbit_{\bflambda}$, let $V\in\Fl_n(\mathbb{C})$ denote the complete flag corresponding to $L$ under the isomorphism \eqref{orbit_complete_isomorphism_intro}. Let $\Orbit_{\bflambda}'$ denote the subset of $\Orbit_{\bflambda}$ of all $L$ such that $\Delta_I(\twist(V))\neq 0$ for all $I\subseteq [n]$. Then we have the continuous map
\begin{align}\label{moment_embedding_permutohedron_product_formula}
\Orbit_{\bflambda}'\to\prod_{k=1}^{n-1}\Perm{\bflambda^{(k)}}, \quad L \mapsto \Big(\mu\big(\Moserextend_{\bflambda^{(k)}}\big(\Delta_I(\twist(V)) : I\in\textstyle\binom{[n]}{k}\big)\big)\Big)_{1 \le k \le n-1}.
\end{align}
Moreover, when \eqref{moment_embedding_permutohedron_product_formula} is restricted to any subset of $\Orbit_{\bflambda}'$ which fixes the complex argument of $\Delta_I(\twist(V))$ for each nonempty $I\subsetneq [n]$, then it is a diffeomorphism onto its image. In particular, \eqref{moment_embedding_permutohedron_product_formula} restricted to $\Orbit_{\bflambda}^{>0}$ is a diffeomorphism onto its image.
\end{prop}

Note that \eqref{moment_embedding_permutohedron_product_formula} is obtain by applying \eqref{embedding_remark_formula} (to $L = \ii M$, rather than $M$), and then applying $\mu$ to each of the $n-1$ factors.
\begin{proof}
This follows from \cref{jacobi_manifold}, using the discussion in \cref{embedding_remark}.
\end{proof}

\begin{rmk}\label{moment_product_extend_remark}
It is not clear how to extend \cref{moment_embedding_permutohedron_product} from $\Orbit_{\bflambda}^{>0}$ to $\Orbit_{\bflambda}^{\ge 0}$. This is because the Moser map $\Moser_{\bflambda}$ is only defined on $\mathbb{P}^{n-1}_{>0}$, not $\mathbb{P}^{n-1}_{\ge 0}$.
\end{rmk}

\begin{rmk}\label{moment_map_image_remark}
We note that in \eqref{moment_embedding_permutohedron_product_formula}, the codomain has dimension much greater than that of the domain ($\prod_{k=1}^{n-1}(\binom{n}{k}-1)$ versus $n!$). This comes from the same property of the Pl\"{u}cker embedding \eqref{plucker_embedding_flag}, whose image in $\prod_{k=1}^{n-1}\mathbb{P}^{\left(\hspace*{-1pt}\binom{n}{k}-1\right)}(\mathbb{C})$ is cut out by {\itshape Grassmann-Pl\"{u}cker} and {\itshape incidence-Pl\"{u}cker relations}. It may be interesting to study the image of this subset in $\prod_{k=1}^{n-1}\Perm{\bflambda^{(k)}}$.
\end{rmk}

The formula \eqref{moment_map_explicit_formula} suggests an interpolation between the two maps considered in \cref{twisted_moment_embedding} and \cref{moment_embedding_permutohedron_product}, namely,
\begin{align}\label{interpolated_moment_map}
\Orbit_{\bflambda} \to \prod_{k=1}^{n-1}\hypersimplex_{k,n}, \quad L \mapsto \big(\mu_{k,n}(\twist(V_k))\big)_{1 \le k \le n-1}
\end{align}
(with notation as in \cref{moment_map_explicit_remark}). We do not expect \eqref{interpolated_moment_map} to be injective on $\Orbit_{\bflambda}$, for similar reasons as discussed in \cref{embedding_remark}. However, we do not know whether it is injective when we require all Pl\"{u}cker coordinates to be nonnegative:
\begin{prob}\label{interpolated_moment_map_problem}
Is the map \eqref{interpolated_moment_map} injective when restricted to $\Orbit_{\bflambda}^{\ge 0}$?
\end{prob}

\bibliographystyle{alpha}
\bibliography{ref}

\end{document}